\documentclass[3p,12pt,fleqn]{elsarticle}

\usepackage{amssymb,amsmath,amsthm,bm}
\usepackage{amsfonts}
\usepackage{epsfig}
\usepackage{amsbsy}
\usepackage{color}
\biboptions{sort&compress}

\usepackage{graphicx}
\graphicspath{{figures/}}

\newlength\imagewidth
 \newtheorem{lemma}{Lemma}[section]
 \newtheorem{theorem}{Theorem}[section]

\usepackage{color}
\definecolor{dgreen}{rgb}{0,.6,0}
\usepackage[normalem]{ulem}

\begin{document}

\begin{frontmatter}


\title{Construction of modulated amplitude waves via averaging in collisionally inhomogeneous Bose-Einstein condensates
}

\medskip
\author[Rome,Paestum]{Qihuai Liu\corref{cor1}}
\ead{qhuailiu@gmail.com}
\author[Rome]{Dingbian Qian}
\ead{dbqian@suda.edu.cn} \cortext[cor1]{Corresponding author at:
School of Mathematics and Computing Science, Guilin University of
Electronic Technology, No. 2, Jinji Street, Guilin 541004, China.
Tel./Fax.: +086 0773 3939803. }

\address[Rome]{School of Mathematics and Computing Science, Guilin University of Electronic Technology,
Guilin 541002, China}
\address[Paestum]{School  of Mathematical Sciences, Soochow  University, Suzhou 215006,  China}

\begin{abstract}
We apply the averaging method to analyze spatio-temportal
structures in nonlinear Schr\"{o}dinger equations and thereby
study the dynamics of quasi-one-dimensional collisionally
inhomogeneous Bose-Einstein condensates with the scattering length
varying periodically in spatial and crossing zero. Infinitely many
(positive measure set) modulated amplitude waves (periodic and
quasi-periodic), which are instable, can be proved to exist by
adjusting the integration constant $c$ on some open interval.
Finally, some numerical simulations support our results.


\end{abstract}

\begin{keyword}
Modulated amplitude waves; Gross-Pitaevskii equations;
Collisionally inhomogeneous Bose-Einstein condensates; Averaging
method




\PACS: 05.45.-a\sep 03.75.Lm\sep05.30.Jp\sep 05.45.Ac
\end{keyword}

\end{frontmatter}

\section{Introduction}
\label{Sec 1} Since the experimental realization of Bose-Einstein
condensates (BECs) in the mid-1990s
\cite{anderson1995observation,davis1995bose}, the study of
matter-wave patterns including existence and stability in BECs has
drawn a great deal of interest from experimentalists
\cite{strecker2002formation,strecker2003bright} and theorists
\cite{kevrekidis2003feshbach,porter2005bose,kapitula1998stabilityb,zharnitsky2005averaging}.

In atomic physics, the Feshbach resonance of the scattering length
of interatomic interactions is used for control of Bose-Einstein
condensates \cite{cornish2000stable,inouye1998observation}. We
consider the main model of this paper for Feshbach resonance given
by the perturbed Gross-Pitaevskii (GP) equation of the
dimensionless form \cite{porter2007modulated}
\begin{equation}\label{1.1}
 i\frac{\partial\psi}{\partial t}=-\frac{1}{2}\frac{\partial^2\psi}{\partial
 x^2}+\tilde{g}(x)|\psi|^2\psi+\tilde{V}(x)\psi,
\end{equation}
where the nonlinearity coefficient $\tilde{g}(x)$ varies in space.
In Eq.~(\ref{1.1}), $\psi$ is the mean-field condensate wave
function (with density  $|\psi|^2$ measured in units of the peak
1D density $n_0$), $x$ and $t$ are normalized, respectively, to
the healing length $\xi=\hbar/\sqrt{n_0|g_1|m}$ and $\xi/c$ (where
$c=\hbar\sqrt{n_0|g_1|/m}$ is the Bogoliubov speed of sound), and
energy is measured in units of the chemical potential $\delta=
g_1n_0$. In the above expressions, $g_1=2\hbar\omega_\perp a_0$,
where $\omega_\perp$ denotes the confining frequency in the
transverse direction, and $a_0$ is a characteristic (constant)
value of the scattering length relatively close to the Feshbach
resonance. Finally, $\tilde{V}(x)$ is the rescaled external
trapping potential, and the $x$-dependent nonlinearity is given by
$\tilde{g}(x)=a(x)/a_0$, where $a(x)$ is the spatially varying
scattering length.

In the past few years, GP equation (\ref{1.1}) has been widely
studied, such as the stability and dynamics of bright, dark
solitary waves
\cite{rapti2007solitary,theocharis2005matter,pelinovsky2003averaging,pelinovsky2004hamiltonian}
and modulated amplitude waves (MAWs) \cite{porter2007modulated}.

In order to study the dynamics of BECs with scattering length
subjected to a spatially periodic variation, Porter and Kevrekidis
et al. \cite{porter2007modulated} transform equation (\ref{1.1})
into a new GP equation with a constant coefficient and an
additional effective potential
\begin{align*}
 i\frac{\partial\psi}{\partial t}&=-\frac{1}{2}\frac{\partial^2\psi}{\partial
 x^2}+|\psi|^2\psi+\tilde{V}(x)\psi+\tilde{V}_{eff}(x)\psi,\\
\tilde{V}_{eff}(x)&=\frac{1}{2}\frac{f''}{f}-\frac{(f')^2}{f^2}+\frac{f'}{f}\frac{\partial}{\partial
x}
\end{align*}
with $f(x)=\sqrt{\tilde{g}(x)}$, then the transformed equation was
investigated. For weak underlying inhomogeneity, the effective
potential takes a form resembling a superlattice, and the
amplitude dynamics of the solutions of the constant-coefficient GP
equation obey a nonlinear generalization of the Ince equation. In
the small-amplitude limit, they use averaging to construct
analytical solutions for modulated amplitude waves (MAWs), whose
stability was subsequently examined using both numerical
simulations of the original GP equation and fixed-point
computations with the MAWs as numerically exact solutions.
However, mentioned in their paper, the transformation
$\tilde{\psi}=\sqrt{\tilde{g}(x)}\psi$ applies only in the case
when $\tilde{g}(x)$ does not cross zero. A natural question is
that for general periodic function $\tilde{g}(x)$, whether the
similar results upon the dynamics can be obtained. On the other
hand, the phases of MAWs considered in \cite{porter2007modulated}
are trivial, which are corresponding to standing waves. Thus,
another question is that whether the MAWs with nontrivial phases
can exist.

With these questions discussed above, in this paper we investigate
the existence and stability of MAWs with nontrivial phases in
collisionally inhomogeneous BECs modeled by GP equation
(\ref{1.1}) for general small periodic function $\tilde{g}(x)$.
The method is based on averaging, and we use the averaging
principle to replace a GP equation by the corresponding averaged
system. Along this paper, we assume that $g(x)$ and $V(x)$ are
analytic and periodic functions with the least positive period
$T=\pi/\sqrt{\delta}$.

The rest paper is organized as follows. In Section \ref{SEC 2}, we
introduce modulated amplitude waves involving periodic and
quasi-periodic, and an averaging theorem is obtained in Section
\ref{Sec 3}. In Section \ref{Sec 4} we investigate the existence
and stability of equilibrium points for the averaged system and
thereby study the periodic orbits and  a numerical simulation is
presented as prescribed parameters in Section \ref{Sec 5}.
Finally, we summarize our results in Section \ref{Sec 6}.

\section{Coherent structure}\label{SEC 2}\setcounter{equation}{0}

we consider uniformly propagating coherent structures with the
ansatz
\begin{equation}\label{2.1}
\psi(t,x)=R(x)\exp(i[\Theta(x)-\mu t]),
\end{equation}
where $R(x)\in\mathbb{R}$ gives the amplitude dynamics of the
condensate wave function, $\theta(x)$ determines the phase
dynamics, and the ``chemical potential'' $\mu$, defined as the
energy it takes to add one more particle to the system, is
proportional to the number of atoms trapped in the condensate.
When the (temporally periodic) coherent structure (\ref{2.1}) is
also spatially periodic, it is called a \emph{modulated amplitude
wave} (MAW) \cite{brusch2000modulated,brusch2001modulated}.
Similarly, a solution of the equation (\ref{1.1}) with the
(temporally periodic) coherent structure (\ref{2.1}) is called a
\emph{quasi-periodic modulated amplitude wave} (QMAW) if it is
also spatially quasi-periodic.

Inserting (\ref{2.1}) into (\ref{1.1}), we obtain the following
two couple nonlinear ordinary differential equations
\begin{equation}\label{3}
R''+\delta R-\frac{c^2}{R^3}+\varepsilon g(x) R^3+\varepsilon
V(x)R=0,
\end{equation}
\begin{equation}\label{4}
\Theta''+2\Theta'R'/R=0~\Rightarrow~\Theta'(x)=\frac{c}{R^2},\quad\quad
\end{equation}
where
\begin{equation*}
    \varepsilon g(x):=-\tilde{g}(x),~\varepsilon
    V(x):=-\tilde{V}(x)
\end{equation*}
and the integration constant $c$, determined by the velocity and
number density, plays the role of ``angular momentum''
\cite{bronski2001bose}.

In case of $c=0$, the phase of the condensate wave function
(standing wave) is trivial and constant. In the general case,
$c\neq 0$, the system (\ref{3}) becomes more
 complicated and the phase is no longer constant
 \cite{chong2004spatial}. Even the amplitude $R(x)$, a solution of
 (\ref{3}), is $T$-periodic, the corresponding condensate wave
 function $\psi(x,t)$ may be not periodic, but quasi-periodic, with respect to the
 spatial variable $x$ \cite{liu2011}.

\section{Averaging theorem}\label{Sec
3}\setcounter{equation}{0}

Rewrite equation (\ref{3}) in the planar equivalent form
\begin{equation}\label{3.1}
\left\{ \begin{array}{ll}
R'=S\\
S'=-\delta R+\displaystyle\frac{c^2}{R^3}-\varepsilon g(x)
R^3-\varepsilon V(x)R.
\end{array}
\right.
\end{equation}
Generally, averaging method involves two steps: transforming to
standard form; solving the averaging equation. In order to proceed
we need to transform (\ref{3.1}) to a standard form for the method
of averaging.

\begin{lemma}\label{LM 3.1} Under the transformation $\Psi:\mathbb{T}\times
\big(\sqrt[4]{\frac{c^2}{\delta}},+\infty \big)\rightarrow
(0,+\infty)\times \mathbb{R}$ defined by
\begin{equation*}
\left\{ \begin{array}{ll} R
=\rho\sqrt{\cos^2(\sqrt{\delta}x+\theta)+\displaystyle\frac{c^2}{\delta\rho^4}\sin^2(\sqrt{\delta}x+\theta)}\\[2em]
S
=\rho\sqrt{\delta}\left(\displaystyle\frac{c^2}{\delta\rho^4}-1\right)\displaystyle\frac{\cos(\sqrt{\delta}x+\theta)\sin(\sqrt{\delta}x+\theta)}
{\sqrt{\cos^2(\sqrt{\delta}x+\theta)+\displaystyle\frac{c^2}{\delta\rho^4}\sin^2(\sqrt{\delta}x+\theta)}}
,
\end{array}
\right.
\end{equation*}
system \emph{(\ref{3.1})} changes into a new system
\begin{equation}\label{3.2}
\left\{ \begin{array}{llll} \rho'
=\varepsilon\left\{\displaystyle\frac{g(x)}{\sqrt{\delta}}\rho^3\left[\frac{1}{4}(1+\frac{c^2}{\delta\rho^4})\sin2(\sqrt{\delta}x+\theta)
+\frac{1}{8}(1-\frac{c^2}{\delta\rho^4})\sin4(\sqrt{\delta}x+\theta)\right]\right.\\
~~~~~~~~~~~~~~~~~~~~~~~~~~~~~~~~~~~~~+\left.\displaystyle\frac{\rho}{2\sqrt{\delta}}V(x)\sin2(\sqrt{\delta}x+\theta)\right\}\\[2em]
\theta'
=\varepsilon\left\{\displaystyle\frac{g(x)(\delta\rho^4+c^2)}{8\delta^2\rho^2}\big(3+\cos4(\sqrt{\delta}x+\theta)\big)+
\displaystyle\frac{g(x)(\delta^2\rho^8+c^4)}{2\delta^2\rho^2(\delta\rho^4-c^2)}\cos2(\sqrt{\delta}x+\theta) \right.\\[1em]
~~~~~~~~~~~~~~~~~~~~~~~~~~~~~~~~~~~~~+\left.\displaystyle\frac{1}{2\delta}V(x)\left(1+\frac{\delta\rho^4+c^2}{\delta\rho^4-c^2}\cos2(\sqrt{\delta}x+\theta)\right)\right\}
\end{array}
\right.
\end{equation}
with the new coordinates $(\theta,\rho)$ in the half-plane
$\mathbb{T}\times \big(\sqrt[4]{\frac{c^2}{\delta}},+\infty
\big)$.
\end{lemma}

The transformation $\Psi$ arises from the variation of constant by
using the solutions of the unperturbed system ($\varepsilon=0$),
and $\rho$ plays the role of ``energy". When taking the
integration constant $c=0$, the transformation $\Psi$ is the usual
change of polar coordinates in the half plane. The proof the Lemma
\ref{LM 3.1} follows from the basic computation (maybe lengthy),
and it can be found in \cite{liu2011}.

 Now write the $T$-periodic
functions $g(x),V(x)$ as the Fourier series
\begin{align}
   \label{3.3} g(x)&=g_0+\sum_{k=1}^{\infty}(\alpha_k \sin
    2k\sqrt{\delta}x+\beta_k \cos
    2k\sqrt{\delta}x),\\
   \label{3.4}     V(x)&=v_0+\sum_{k=1}^{\infty}(a_k \sin
    2k\sqrt{\delta}x+b_k \cos
    2k\sqrt{\delta}x).
\end{align}
After inserting (\ref{3.3}) and (\ref{3.4}) into (\ref{3.2}) and
then multiplying the right side of (\ref{3.2}) by $1/T$ and
integrating from $0$ to $T$, we obtain the averaged system
\begin{equation}\label{3.5}
\left\{ \begin{array}{llll} \rho'
=\varepsilon\left\{\displaystyle\frac{\delta
\bar{\rho}^4+c^2}{8\delta\sqrt{\delta}\bar{\rho}}A\sin(2\bar{\theta}+\phi_1)
 +\displaystyle\frac{\bar{\rho}}{4\sqrt{\delta}}B\sin(2\bar{\theta}+\phi_2)\right\}:=\varepsilon
F_1(\bar{\theta},\bar{\rho})\\[2em]
\theta'
=\varepsilon\left\{\displaystyle\frac{3g_0(\delta\bar{\rho}^4+c^2)}{8\delta^2\bar{\rho}^2}
+
\displaystyle\frac{(\delta^2\bar{\rho}^8+c^4)}{4\delta^2\bar{\rho}^2(\delta\bar{\rho}^4-c^2)}A\cos(2\bar{\theta}+\phi_1)\right.\\[1em]
~~~~~~~~~~~~~~~~~~~~~~~~~+\left.\displaystyle\frac{v_0}{2\delta}+
\displaystyle\frac{1}{4\delta}\frac{\delta\bar{\rho}^4+c^2}{\delta\bar{\rho}^4-c^2}B\cos(2\bar{\theta}+\phi_2)\right\}:=\varepsilon
F_2(\bar{\theta},\bar{\rho}),
\end{array}
\right.
\end{equation}
where
\begin{align*}
    A&=\sqrt{\alpha_1^2+\beta_1^2},~~~B=\sqrt{a_1^2+b_1^2},\\
    \phi_1&=\arctan\frac{\alpha_1}{\beta_1},~(\phi_1=\frac{\pi}{2}\cdot\mathrm{sign}(\phi_1),~\text{if}~\beta_1=0),\\
    \phi_1&=\arctan\frac{a_1}{b_1},~(\phi_1=\frac{\pi}{2}\cdot\mathrm{sign}(a_1),~\text{if}~b_1=0).
\end{align*}

\begin{theorem}\label{TH 3.1}
\textbf{\emph{[~Averaging theorem~]}} There exists a $c^r,r\geq2$,
change of variables
\begin{equation*}
    \rho=\bar{\rho}+\varepsilon
    w_1(\bar{\theta},\bar{\rho},x,\varepsilon),~~  \theta=\bar{\theta}+\varepsilon w_2(\bar{\theta},\bar{\rho},x,\varepsilon)
\end{equation*}
with $w_1,w_2$ $T$-periodic functions of $x$, transforming
\emph{(\ref{3.2})} into
\begin{equation}\label{3.6}
\left\{ \begin{array}{ll} \bar{\rho}'=\varepsilon
F_1(\bar{\theta},\bar{\rho})+\varepsilon^2
    g_1(\bar{\theta},\bar{\rho},x,\varepsilon)\\[0.5em]
\bar{\theta}'=\varepsilon F_2(\bar{\theta},\bar{\rho})
+\varepsilon^2
    g_2(\bar{\theta},\bar{\rho},x,\varepsilon)
\end{array}
\right.
\end{equation}
with $g_1,g_2$ $T$-periodic functions of $x$. Moreover,

\emph{(i)} If $(\theta_\varepsilon(x),\rho_\varepsilon(x))$ and
$(\bar{\theta}(x),\bar{\rho}(x))$ are solutions of the original
system \emph{(\ref{3.2})} and averaged system \emph{(\ref{3.5})}
respectively, with the initial value such that
\begin{equation*}
    |\,\rho_\varepsilon(0)-\bar{\rho}(0)|+|\,\theta_\varepsilon(0)-\bar{\theta}(0)|=\mathcal{O}(\varepsilon),
\end{equation*}
then
\begin{equation*}
    |\,\rho_\varepsilon(x)-\rho_0(x)|+|\,\theta_\varepsilon(x)-\theta_0(x)|=\mathcal{O}(\varepsilon),
\end{equation*}
for times $x$ of order $1/\varepsilon$.

\emph{(ii)} If  $P_0$ is an equilibrium point of
\emph{(\ref{3.5})} such that the corresponding Jacobian matrix has
no eigenvalue equal to zero, then \emph{(\ref{3.2})} admits a
$T$-periodic solution
$(\theta_\varepsilon(t),\rho_\varepsilon(t))$ such that
    $|(\rho_\varepsilon(t),\theta_\varepsilon(t))-P_0|=\mathcal{O}(\varepsilon),$
for sufficiently small $\varepsilon$; if  $P_0$ is a hyperbolic
equilibrium point of \emph{(\ref{3.5})}, then there exists
$\epsilon_0>0$  such that, for all $0<\varepsilon<\varepsilon_0$,
system (\ref{3.2}) possesses a hyperbolic periodic orbits
$\gamma_\varepsilon(x)=P_0+\mathcal{O}(\varepsilon)$  of the same
stability type as $P_0$.

\emph{(iii)}~If~$(\theta_\varepsilon(x),\rho_\varepsilon(x))\in
W^s(\gamma_\varepsilon(x))$~is a solution of system (\ref{3.2})
lying in the stable manifold of the hyperbolic periodic orbit
$\gamma_\varepsilon(x)=P_0+\mathcal{O}(\varepsilon)$,
$(\bar{\theta}(x),\bar{\rho}(x))\in W^s(P_0)$ is a solution of
system (\ref{3.5}) lying in the stable manifold of the hyperbolic
equilibrium point $P_0$ and
 \begin{equation*}
    |(\theta_\varepsilon(0)),\rho_\varepsilon(0)-(\bar{\theta}(0),\bar{\rho}(0))|=\mathcal{O}(\varepsilon),
\end{equation*}
then
 \begin{equation*}
    |(\theta_\varepsilon(x),\rho_\varepsilon(x))-(\bar{\theta}(x),\bar{\rho}(x))|=\mathcal{O}(\varepsilon),
\end{equation*}
for $x\in [0,+\infty)$. Similar results apply to solutions
 lying in the instable manifold on the interval $x\in(-\infty,0]$.

\emph{(iv)} If $P_0$ is an equilibrium point of system
\emph{(\ref{3.5})} and there exists a neighborhood $U(r;P_0)$ of
$P_0$ (with radius $r$ and center $P_0$) such that there is not
another equilibrium point in the closure of $U$ and
\begin{equation*}
    \mathrm{deg}(F,U,P_0)\neq 0
\end{equation*}
with $F=(F_1,F_2)$. Then for $|\varepsilon|>0$ sufficiently small,
there exist a $T$-periodic solution $\varphi_\varepsilon(x)$ of
system (\ref{3.2}) such that
\begin{equation*}
\varphi_\varepsilon(\cdot,\varepsilon)\rightarrow
P_0£¬~~as~\varepsilon\rightarrow 0.
\end{equation*}
\end{theorem}

\begin{proof}
 The proof of (i)-(iii) follows directly from \cite{berglund2001perturbation} or
\cite{guckenheimer1983nonlinear}. The proof of (iv) is based a
framework of coincidence degree theory. Without loss of
generality, we assume $P_0=0$. We define homotpoy operator
$\mathcal{H}:C([0,T],\mathbb{R}^2)\times[0,1]\rightarrow
L_1([0,T],\mathbb{R}^2)$ by
\begin{equation*}
    \mathcal{H}(\bar{\theta},\bar{\rho},\lambda):=\lambda\varepsilon
    F(\bar{\theta},\bar{\rho})+(1-\lambda)\varepsilon^2
    G(\bar{\theta},\bar{\rho},x,\varepsilon),
\end{equation*}
where $G:=(g_1,g_2)$. According to \cite[Ch.
VI]{mawhintopological} , $\mathcal{H}$ is $L$-compact on
$\Omega\times[0,T]$, where $\Omega$ is a bounded open set of
$C([0,T],\mathbb{R}^2)$ defined by
\begin{equation*}
   \Omega:=\{(\bar{\theta},\bar{\rho})\in
   C([0,T],\mathbb{R}^2):\|(\bar{\theta},\bar{\rho})\|<r\}.
\end{equation*}
We remark that $(\bar{\theta},\bar{\rho})\in\Omega$ is a
$T$-periodic solution of system (\ref{3.6}) if and only if
$(\bar{\theta},\bar{\rho})$ is a  solution of
$\mathcal{L}(\bar{\theta},\bar{\rho})=\mathcal{H}(\bar{\theta},\bar{\rho},0)$
in $\overline{\Omega}$. Since there is not another equilibrium
point in the closure of $U$, we let
\begin{align*}
    &M_1=\min_{(\bar{\theta},\bar{\rho})\in \partial
    U}|F_1(\bar{\theta},\bar{\rho})-F_2(\bar{\theta},\bar{\rho})|>0,\\
    &M_2(\varepsilon)=\varepsilon\max \{|g_1(\bar{\theta},\bar{\rho},x,\varepsilon)-g_1(\bar{\theta},\bar{\rho},x,\varepsilon)|:(\bar{\theta},\bar{\rho},x,\varepsilon)\in
    \overline{U}\times[0,T]\}
\end{align*}
with $M_2(\varepsilon)\rightarrow0$ as $\varepsilon\rightarrow0$.
We also assume $M_2(\varepsilon)>0$, otherwise $P_0$ is a solution
of system (\ref{3.6}) and the result is proved.

 First, we claim that, for each $\varepsilon\in (-\varepsilon_0,0)\cup(0,\varepsilon_0)$ with $\varepsilon_0=\max_{\varepsilon\in[-1,1]}\{M_1/M_2(\varepsilon)\}$, there exists no solution
$(\bar{\theta},\bar{\rho})\in
\partial\Omega$ for the operator equation
\begin{equation}\label{3.7}
    \mathcal{L}(\bar{\theta},\bar{\rho})=\mathcal{H}(\bar{\theta},\bar{\rho},\lambda),~~\lambda\in(0,1].
\end{equation}
In fact, if $(\bar{\theta},\bar{\rho})\in\partial\Omega$ is a
solution of (\ref{3.7}), then there exists $\xi\in[0,T]$ such that
\begin{align*}
   &\bar{\theta}'(\xi)+\bar{\rho}'(\xi)=0,\\
&\|(\bar{\theta},\bar{\rho})\|=\max_{x\in[0,T]}\sqrt{\bar{\theta}^2(x)+\bar{\rho}^2(x)}=\sqrt{\bar{\theta}^2(\xi)+\bar{\rho}^2(\xi)}
\end{align*}
and
\begin{equation*}
    \bar{\theta}'(\xi)\bar{\rho}'(\xi)\leq0.
\end{equation*}
 Thus, it follows that
\begin{align*}
    0&=|\bar{\theta}'(\xi)+\bar{\rho}'(\xi)|\\
    &\geq|\varepsilon|\cdot|
F_1(\bar{\theta},\bar{\rho})-F_2(\bar{\theta},\bar{\rho})|-\varepsilon^2|g_1(\bar{\theta},\bar{\rho},x,\varepsilon)-g_2(\bar{\theta},\bar{\rho},x,\varepsilon)|\\
&\geq|\varepsilon|M_1-\varepsilon M_2(\varepsilon)>0,
\end{align*}
which is a contradiction.

Without loss of generality, we suppose that
\begin{equation}\label{3.8}
    \mathcal{L}(\bar{\theta},\bar{\rho})\neq\mathcal{H}(\bar{\theta},\bar{\rho},\lambda),~~(\bar{\theta},\bar{\rho})\in\partial\Omega
\end{equation}
  holds for
$\lambda\in[0,1]$. Otherwise, the result is proved for
$(\bar{\theta},\bar{\rho})\in\partial\Omega$. Thus, we can apply
the homotopy property of the coincidence degree and obtain
\begin{align*}
    |D_{\mathcal{L}}(\mathcal{L}-\mathcal{H}(\cdot,0),\Omega)|=|D_{\mathcal{L}}(\mathcal{L}-\mathcal{H}(\cdot,1),\Omega)|\\
    =|\mathrm{deg}(F,\Omega\cap\mathbb{R}^2,0)|\neq0.
\end{align*}
Hence, by the existence property of the coincidence degree, there
is $(\bar{\theta},\bar{\rho})\in\overline{\Omega}$ such that
$\mathcal{L}(\bar{\theta},\bar{\rho})=\mathcal{H}(\bar{\theta},\bar{\rho},0).$
Then $(\bar{\theta},\bar{\rho})$ is a $T$-periodic solution of
(\ref{3.6}). Thus, owing to the change of variables in this
theorem, there is a $T$-periodic solution $(\theta,\rho)$ for
system (\ref{3.5}).
\end{proof}

We remark that the proof of part (iv) of Theorem \ref{TH 3.1} does
not need the smoothness condition upon $F$. So, it is convenient
to deal with the existence of periodic solutions for nonlinear
systems with loss of smoothness by Theorem \ref{TH 3.1}.

\section{Periodic orbits and stability}\label{Sec
4}\setcounter{equation}{0}

To study the dynamics of MAWs or QMAWs for system (\ref{1.1}), we
must investigate the behavior of the periodic orbits for system
(\ref{3.1}) including the existence and stability. According to
the method of averaging, the equilibrium point of the averaged
system determines the properties of the periodic orbit of the
corresponding perturbed system. For example, the equilibrium point
with its eigenvalue of linearization nonzero implies that there
exists at least one periodic orbit; in addition, if the
equilibrium point is hyperbolic, then the periodic orbit has the
same type of stability as the equilibrium point, for sufficiently
small parameter $\varepsilon$.

In order to find periodic orbits of system (\ref{3.2}), it is
sufficient to find equilibrium points of the averaged system
(\ref{3.5}). In the following, we will discuss the existence and
stability of the equilibrium points for the averaged system
(\ref{3.5}). For simplification, we assume that
\begin{equation*}
    \phi_2=\pi+\phi_1,~~\phi_1=-\pi/2,~~g_0=A/3,~~v_0=-B/2.
\end{equation*}
Recalling system (\ref{3.5}), together with the assumption, we
have the averaged system
\begin{equation}\label{4.1}
\left\{ \begin{array}{llll} \rho'
=\varepsilon\left\{\displaystyle\frac{A}{8\sqrt{\delta}\bar{\rho}}\Big(\bar{\rho}^4-2\frac{B}{A}\bar{\rho}^2+\frac{c^2}{\delta}\Big)\sin(2\bar{\theta}+\phi_1)
\right\}\\[2em]
\theta' =\varepsilon\left\{\displaystyle\frac{1}{4\delta
\bar{\rho}^2\cdot(\delta\bar{\rho}^4-\frac{c^2}{\delta})}\left(\frac{A}{2}(\bar{\rho}^8-\frac{c^4}{\delta^2})
-B\bar{\rho}^2(\bar{\rho}^4-\frac{c^2}{\delta})\right.\right. \\
\left.\left.~~~~~~~~~~~~~~~~~~~~~~~~~~~~~~~~+\Big[A\big(\bar{\rho}^8+
\displaystyle\frac{c^4}{\delta^2}\big)-B\bar{\rho}^2\big(\bar{\rho}^4+\displaystyle\frac{c^2}{\delta}\big)\Big]\cos(2\bar{\theta}+\phi_2)\right)\right\}.
\end{array}
\right.
\end{equation}

Notice that there exists a constant $c_0>0$ such that for each
$c\in(0,c_0)$, the equation
\begin{equation}\label{4.2}
\bar{\rho}^4-2\frac{B}{A}\bar{\rho}^2+\frac{c^2}{\delta}=0
\end{equation}
has at least two real roots
\begin{equation*}
    \rho_{1,2}=\pm\sqrt{\frac{B}{A}+\sqrt{\frac{B^2}{A^2}-\frac{c^2}{\delta}}}\in(-\infty,-\sqrt[4]{c^2/\delta})\cup(\sqrt[4]{c^2/\delta},+\infty).
\end{equation*}
Moreover, equation (\ref{4.2}) implies that
\begin{equation*}
\frac{A}{2}(\bar{\rho}^8-\frac{c^4}{\delta^2})
-B\bar{\rho}^2(\bar{\rho}^4-\frac{c^2}{\delta})=0.
\end{equation*}
Thus, we can find four equilibrium points as follows
\begin{align*}
    P_{1,2}&:\Big(\pm\sqrt{\frac{B}{A}+\sqrt{\frac{B^2}{A^2}-\frac{c^2}{\delta}}},~k\pi+\frac{\pi}{4}-\frac{\phi_1}{2}\Big),\\
     P_{3,4}&:\Big(\pm\sqrt{\frac{B}{A}+\sqrt{\frac{B^2}{A^2}-\frac{c^2}{\delta}}},~k\pi-\frac{\pi}{4}-\frac{\phi_1}{2}\Big),~k\in\mathbb{Z}.
\end{align*}
The eigenvalues of the equilibrium points $P_{1,2}$ and $ P_{3,4}$
are given by
\begin{equation*}
    \lambda_{1,2}^{(1)}=\frac{\varepsilon
    A}{2\sqrt{\delta}}\sqrt{\frac{B^2}{A^2}-\frac{c^2}{\delta}}>0,~~\lambda_{1,2}^{(2)}=-\frac{\varepsilon
    A(\rho_{1,2}^4-\frac{c^2}{\delta})}{4\delta\rho_{1,2}^2}<0
\end{equation*}
and
$\lambda_{3,4}^{(1)}=-\lambda_{1,2}^{(1)},~\lambda_{3,4}^{(2)}=-\lambda_{1,2}^{(2)}$,
respectively. So, the equilibrium points $P_{1,2}$ and $ P_{3,4}$
are hyperbolic, and as a consequence persist as periodic orbits
for system (\ref{3.2}); in addition, these periodic orbits are
instable.

If $\bar{\theta}=k\pi-\phi_1/2$, to find the equilibrium points,
one will solve the following algebraic equation
\begin{equation}\label{4.3}
    \frac{A}{2}(\bar{\rho}^8-\frac{c^4}{\delta^2})
-B\bar{\rho}^2(\bar{\rho}^4-\frac{c^2}{\delta})+A\big(\bar{\rho}^8+
\displaystyle\frac{c^4}{\delta^2}\big)-B\bar{\rho}^2\big(\bar{\rho}^4+\displaystyle\frac{c^2}{\delta}\big)=0.
\end{equation}
Equation (\ref{4.3}) has at least two roots
\begin{equation*}
    \rho_{3,4}=\pm\sqrt{\frac{4B}{3A}+o(c)}\in(-\infty,-\sqrt[4]{c^2/\delta})\cup(\sqrt[4]{c^2/\delta},+\infty),
\end{equation*}
for sufficiently small integration constant $c>0$.

 \begin{figure}[htbp]
  \begin{center}
      \includegraphics[scale=0.6]{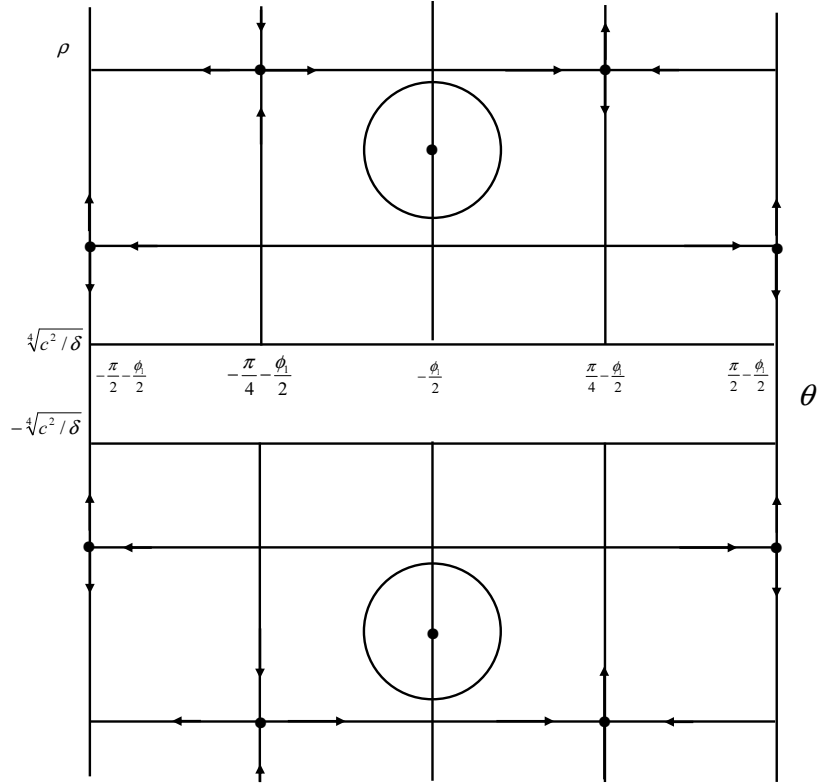}
  \end{center}
\caption{The phase portrait associated with the averaged system
(\ref{4.1}). All the equilibrium points persist as periodic orbits
for system (\ref{3.2}).} \label{Fig 1}
\end{figure}

If $\bar{\theta}=k\pi+\pi/2-\phi_1/2$, the algebraic equation
\begin{equation*}
    f(\rho):=\rho^8-\frac{B}{A}\rho^2+\frac{3c^4}{\delta^2}=0
\end{equation*}
needs to be solved. Note that
\begin{equation}\label{4.4}
f\Big(\pm\sqrt[4]{\frac{c^2}{\delta}}\Big)=4\frac{c^4}{\delta^2}-\frac{B}{A}\big(\frac{c^4}{\delta^2}\big)^{3/2}<0~~\text{and}~~
f(\pm\infty)=+\infty,
\end{equation}
for sufficiently small positive constant $c$. By the mean value
theorem, equation (\ref{4.4}) has two roots $\rho_{5,6}$ such that
\begin{equation*}
    \rho_{5,6}=\pm\sqrt[6]{\frac{c^2B}{\delta A}+o(c^2)}\in(-\infty,-\sqrt[4]{c^2/\delta})\cup(\sqrt[4]{c^2/\delta},+\infty),
\end{equation*}
for sufficiently small $c$. As a consequence, four equilibrium
points of the averaged system (\ref{4.1}) are obtained as follows
\begin{align*}
    P_{5,6}&:\Big(\pm\sqrt{\frac{4B}{3A}+o(c)},~k\pi-\frac{\phi_1}{2}\Big),\\
     P_{7,8}&:\Big(\pm\sqrt[6]{\frac{c^2B}{\delta A}+o(c^2)},~k\pi+\frac{\pi}{2}-\frac{\phi_1}{2}\Big),~k\in\mathbb{Z}.
\end{align*}
The eigenvalues of the linearization at the equilibrium points
$P_{5,6}$ and $ P_{7,8}$ are given by
\begin{equation*}
    \lambda_{5,6}^{(1,2)}= \pm i\varepsilon\sqrt{\frac{4B^2(\frac{8B^2}{9A}-\frac{c^2}{\delta})}{3\delta(\rho_{3,4}^4-\frac{c^2}{\delta})}}
\end{equation*}
and
\begin{equation*}
    \lambda_{7,8}^{(1,2)}= \pm
    \varepsilon \sqrt{\frac{2A^2c^4}
    {4\delta^{7/2}\rho_{5,6}^4(\rho_{5,6}^4-\frac{c^2}{\delta})}\Big(\frac{2B}{A}\rho_{5,6}^2-\rho_{5,6}^4-\frac{c^2}{\delta}\Big)},
\end{equation*}
respectively. The equilibrium points $P_{7,8}$ imply that that two
instable periodic orbit of system (\ref{3.2}) exist; while the
equilibrium points $P_{5,6}$ are nonlinear centers, and also
persist as periodic orbits for system (\ref{3.2}). The phase
portrait for system (\ref{4.1}) is given in Figure 1. Since the
periodic orbits corresponding to the equilibrium points $P_{5,6}$
are not hyperbolic, one can not conclude their stability. This
question is left open for further study.

In summary, for the equilibrium $P_i$ ($i=1,\cdots,4,7,8$), there
exists $\varepsilon_0>0$ such that, for all
$0<\varepsilon\leq\varepsilon_0$, system (\ref{3.2}) possesses a
unique hyperbolic periodic orbits
$\gamma_\varepsilon(x)=P_i+\mathcal{O}(\varepsilon)$, which is
instable.

We also remark that, if
$(\theta_\varepsilon(x),\rho_\varepsilon(x))\in
W_s(\gamma_\varepsilon(x))$ is a solution of system (\ref{3.2})
lying in the stable manifold of the hyperbolic periodic orbit
$\gamma_\varepsilon(x)=P_i+\mathcal{O}(\varepsilon)$,
$(\bar{\theta}(x),\bar{\rho}(x))\in W_s(P_0)$ is a solution of
system (\ref{4.1}) lying in the stable manifold of the hyperbolic
equilibrium $P_i$ ($i=1,\cdots,4,7,8$) and
$|(\theta_\varepsilon(0),\rho_\varepsilon(0))-(\bar{\theta}(0),\bar{\rho})(0)|=\mathcal{O}(\varepsilon)$,
 then
 $|(\theta_\varepsilon(x),\rho_\varepsilon(x))-(\bar{\theta}(x),\bar{\rho})(x)|=\mathcal{O}(\varepsilon)$,
 for all $x\in[0,+\infty)$. Similar results apply to solutions
 lying in the instable manifold on the interval $x\in(-\infty,0]$.

Although we study system (\ref{4.1}) for sufficiently small $c>0$,
$c$ also can be taken on a open interval $(0,\bar{c})$, for some
positive constant $\bar{c}$. By continuous dependence of solutions
with respect to the parameters, there is a connected set
$\mathcal{C}$ of $T$-periodic solutions for system (\ref{3.2}) and
then for system (\ref{3.1}). Since
\begin{eqnarray*}
    \psi(t,x)&=&R(x)\mathrm{exp}{i[\Theta(x)-\mu t]}\\
    &=&R(x)\big(\cos[\bar{\Theta}(x)+\nu x-\mu t]+i\sin[\bar{\Theta}(x)+\nu x-\mu
    t]\big),
\end{eqnarray*}
where
\begin{equation*}
    \nu=\frac{1}{T}\int_{x_0}^{x_0+T}\frac{c}{R^2(\xi)}\mathrm{d}\xi
\end{equation*}
and $\bar{\Theta}(x)=\Theta(x)-\nu$ is a $T$-periodic function
with zero mean value,  whether $\psi(t,x)$ is a MAW or QMAW
depends on the choosing of the integration constant $c$.
Precisely, If $2\pi/\nu$ and $T$ are rationally related, then
$\psi(x,t)$ is a MAW; if $2\pi/\nu$ and $T$ are rationally
irrelevant, then $\psi(x,t)$ is not periodic but quasi-periodic,
which is corresponding to a QMAW with the frequency
$\omega=\langle2\pi/\nu, T\rangle$.

\section{Numerical simulation}\label{Sec 5}\setcounter{equation}{0} To demonstrate the process
of averaging to BECs, a specific example of numerical computation
is given in the following. We take
\begin{align*}
g(x)&=\frac{1}{2}-\frac{3}{2}\sin2x,\\
V(x)&=2\sin2x-1
\end{align*}
and the parameters $\delta=1, c=0.5, \varepsilon=0.01$. Obviously,
$g(x)$ crosses zero. We can find equilibrium points for system
(\ref{4.1}) in the $(\theta,\rho)$-coordinates as follows
\begin{align*}
    P_{1,2}&:(k\pi,~\pm1.60),~~~~~~~~~~
     P_{3,4}:(k\pi+\frac{\pi}{2},~\pm1.60),~~\\
 P_{5,6}&:(k\pi+\frac{3\pi}{4},~\pm1.02),~~
      P_{7,8}:(k\pi+\frac{\pi}{4},~\pm 1.33).
\end{align*}
$P_{7,8}$ are nonlinear centers with eigenvalues of the
linearization $\lambda^{(1,2)}_7=\lambda^{(1,2)}_8=\pm
i\varepsilon$.

 \begin{figure}[htbp]
  \begin{center}
      \includegraphics[scale=0.8]{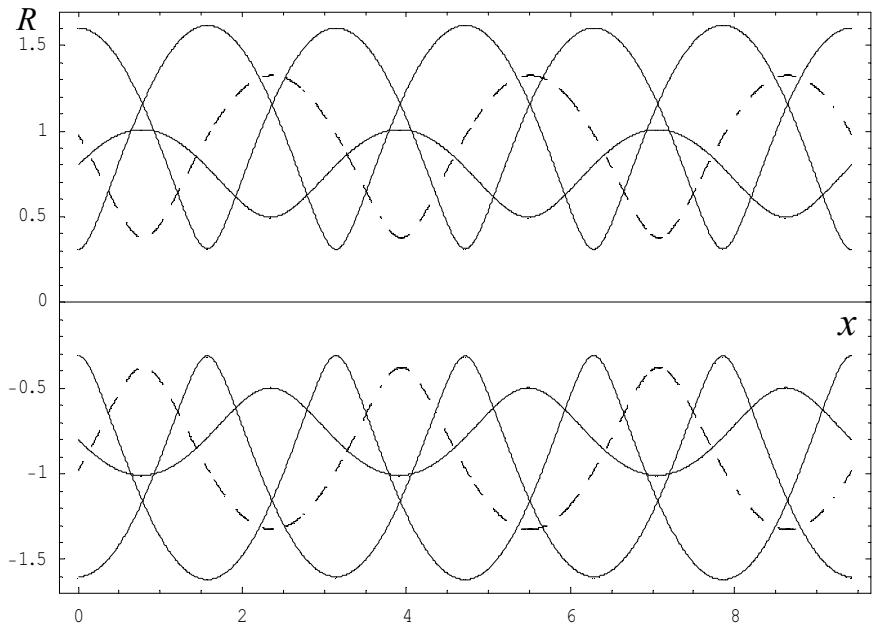}
  \end{center}
\caption{The solutions portrait of system (\ref{3.1}) with initial
value $(R(0),R'(0))=\tilde{P}_i,i=1,2,\cdots,8$. Here, we take
$g(x)=\frac{1}{2}-\frac{3}{2}\sin2x, V(x)=2\sin2x-1$ and the
parameters $\delta=1, c=0.5, \varepsilon=0.01$. These solutions
are a good approximation to the periodic orbits for $x$ of order
$1/\varepsilon$. The solutions with solid lines are instable.}
\label{Fig 2}
\end{figure}

Using the transformation $\Psi$, these equilibrium points in the
$(R,S)$-coordinates with $x=0$ are given by
\begin{align*}
    \tilde{P}_{1,2}&:(\pm1.60,0),~~~~~~~~~~
    \tilde{P}_{3,4}:(\pm0.31,~0),~~\\
 \tilde{P}_{5,6}&:(\pm0.80,~\pm0.48),~~
      \tilde{P}_{7,8}:(\pm 0.98,~\mp0.82).
\end{align*}
We plot the solutions of system (\ref{3.1}) starting from
$\tilde{P}_i,i=1,2,\cdots,8$, according to the averaged theorem,
which are a good approximation to the periodic orbits, see Figure
2.

\section{Conclusion}\label{Sec
6}\setcounter{equation}{0}

In conclusion, we have presented first-order averaging theorem in
the periodic case for dynamics of MAWs (or QMAWs) in collisionally
inhomogeneous BECs. The transformed system is non-Hamiltonian, and
we indicate how the averaging theorems can be used to prove the
existence and stability of periodic solutions. The questions as
mentioned in the introduction have been answered. When the
sufficiently small scattering length $a(x)$ varies periodically in
spatial variable $x$ and crosses zero, infinitely many (positive
measure set) MAWs and QMAWs can be proved to exist by adjusting
the integration constant $c$ on some open interval.

A numerical approximation of periodic orbits is given for some
prescribed parameters. We remark that, expanding at each
equilibrium point and combining with multiple scale perturbed
theory, such as work in
\cite{porter2004perturbative,porter2005bose,porter2004modulated,porter2004resonant},
there may be a better approximation for each continuation periodic
orbit. However, we emphasis on the theory frame of averaging to
 study dynamics of collisionally
inhomogeneous BECs.

In the end, it should be remember that the asymptotic
approximation are valid for small $\varepsilon$, but how small is
usually a difficult problem. However, one advantage of averaging
is obvious that it is set up for an easy return to the original
variables.

\section*{Acknowledgements}
This work is supported by the National Natural Science Foundation
of China (10871142) and Doctoral Fund of Ministry of Education of
China (20070285002).

\bibliographystyle{elsarticle-num-names}

\end{document}